\documentclass[12pt]{article}
\usepackage{fullpage,amsmath, amsfonts, amssymb, amsthm}
\usepackage[round]{natbib}
\usepackage{graphicx, algorithm, algpseudocode, enumerate, hyperref, color, bm, bbm,xcolor}
\usepackage{appendix}
\usepackage{subfig}

\newcommand{\real}{\mathbb R}
\def\Cov{\mathrm{Cov}}
\def\E{\mathrm{E}}

\def\Prob{\mathbb{P}}
\def\swap{\mathrm{swap}}
\def\H{\mathcal{H}}
\def\A{\mathcal{A}}

\def\V{\mathcal{V}}
\def\R{\mathcal{R}}
\def\F{\mathcal{F}}
\def\indi{\mathbbm{1}}

\def\fdp{\textsc{FDP}}
\def\wfdp{\textsc{wFDP}}
\def\U{\mathcal{U}}

\DeclareMathOperator*{\argmax}{arg\,max}
\DeclareMathOperator*{\argmin}{arg\,min}
\newcommand{\norm}[1]{\left\lVert#1\right\rVert}
\newcommand{\snorm}[1]{\|#1\|}

\newtheorem{theorem}{Theorem}
\newtheorem{lemma}[theorem]{Lemma}

\newtheorem{remark}[theorem]{Remark}

\newtheorem{definition}[theorem]{Definition}

\title{Controlling Costs: Feature Selection on a Budget}
\author{Guo Yu\thanks{Department of Statistics and Applied Probability, University of California Santa Barbara,  Santa Barbara, California, 93110, \href{mailto:guoyu@ucsb.edu}{guoyu@ucsb.edu} }  \and Daniela Witten\thanks{Department of Statistics and Biostatistics, University of Washington, Seattle, Washington, 98195, \href{mailto:dwitten@uw.edu}{dwitten@uw.edu} } \and Jacob Bien\thanks{Department of Data Sciences and Operations, Marshall School of Business, University of Southern California, Los Angeles, CA 90089, \href{mailto:jbien@usc.edu}{jbien@usc.edu}}}
\date{}

\begin{document}
\maketitle

\begin{abstract}
  The traditional framework for feature selection treats all features as costing the same amount. However, in reality, a scientist often has considerable discretion regarding which variables to measure, and the decision involves a tradeoff between model accuracy and cost (where cost can refer to money, time, difficulty, or intrusiveness). 
  In particular, unnecessarily including an expensive feature in a model is worse than unnecessarily including a cheap feature. We propose a procedure, which we call cheap knockoffs,
  for performing feature selection in a cost-conscious manner. 
  The key idea behind our method is to force higher cost features to compete with more knockoffs than cheaper features. We derive an upper bound on the weighted false discovery proportion associated with this procedure, which corresponds to the fraction of the feature cost that is wasted on unimportant features.  We prove that this bound holds simultaneously with high probability over a path of selected variable sets of increasing size. A user may thus select a set of features based, for example, on the overall budget, while knowing that no more than a particular fraction of feature cost is wasted.
  We investigate, through simulation and a biomedical application, the practical importance of incorporating cost considerations into the feature selection process.
\end{abstract}

\section{Introduction} \label{sec:introduction}
The traditional framework for feature selection ignores the fact that, in practice,  different features may have different costs.
In reality, practitioners must balance the opposing demands of model accuracy and budget considerations.
For example, as we will see in Section \ref{sec:data}, in medical diagnosis, doctors often have a wide range of options for what features to measure: a laboratory result may provide highly relevant information yet is expensive in terms of money, time, and the burden on patients; a simple questionnaire or even demographic information may be less informative but incurs lower costs. 
When a questionnaire would suffice for forming an accurate diagnosis, performing a laboratory examination would be practically misguided.
Likewise, how should we decide whether to sequence a patient's entire genome or simply to conduct some cheap lab tests?
This same challenge appears in other domains. For example,
to determine the veracity of an online news article, do we require high-quality features based on an expert's reading, or do features derived from natural language processing suffice?

Consider the response of interest $Y$ and a set of features $X_1, \ldots, X_p$, where for each feature $X_j$, there is an associated cost $\omega_j > 0$. In this paper, we consider a very general model where
$Y | X_1, \ldots, X_p$
follows an arbitrary distribution, and we assume that the joint distribution of $X_1, \ldots, X_p$ is known.
Let $\H_0$ be the set of irrelevant features, i.e., $j \in \H_0$ if and only if  $X_j$ is independent of $Y$ conditional on the other variables $\{X_k: k\neq j\}$ \citep[Definition 1 in][]{candes2018panning}.
Given a set of selected features $\R \subseteq \{1, \ldots, p\}$, the false discovery proportion ($\fdp$) is defined as $|\R \cap \H_0| / |\R|$, i.e., it is the fraction of selected features that are unnecessarily included.

\citet{barber2015} proposed the knockoff filter, a feature selection procedure that provably controls the false discovery rate, defined as $\E(\fdp)$. 
For each feature, they construct a knockoff feature, i.e., a carefully constructed fake copy of that feature. 
A feature is then only selected if it shows considerably more association with the response than its knockoff counterpart. 
\citet{katsevich2018towards} showed that one can 
directly upper-bound the false discovery proportion, with high probability, simultaneously for an entire path of selected models, $\R_1, \ldots, \R_p$, where $\R_k \subseteq \R_{k + 1}$ for all $k$.

However, the false discovery proportion and the false discovery rate put all features on an equal footing, and do not consider their costs $\omega_1, \ldots, \omega_p$.
To overcome this shortcoming, the weighted false discovery proportion ($\wfdp$; \citealt{benjamini1997multiple}) is defined as $ \wfdp (\R) = C(\R \cap \H_0) / C(\R)$, i.e., the fraction of the total cost that is wasted, where $C(\A) = \sum_{j \in \A} \omega_j$ is the cost of measuring the features in $\A$.

The weighted false discovery proportion and weighted false discovery rate are not new \citep{benjamini1997multiple,benjamini2007false}, and the Benjamini-Hochberg procedure  \citep{benjamini1995controlling} has been generalized to the weighted false discovery rate setting.
A related criterion is the penalty-weighted false discovery rate \citep{ramdas2019unified}, which can be controlled with the p-filter.
However, the aforementioned procedures only provably control the corresponding criteria under restrictive dependence assumptions on the $p$-values \citep{benjamini2001}. 
Under arbitrary dependence, the reshaping process \citep{benjamini2001, blanchard2008two, ramdas2019unified} needs to be applied, which can greatly reduce power.
\citet{basu2018weighted} proposed a procedure that has asymptotic control of a related quantity,
namely $\E[C(\R \cap \H_0)] / \E[C(\R)]$, in a mixture model under certain regularity conditions.

In this work, we adapt the ideas of knockoffs \citep{barber2015} and simultaneous inference \citep{goeman2011multiple,katsevich2018towards} to the setting where features have costs. 
The key to our method, which we call {\em cheap knockoffs}, is to construct multiple knockoffs for each feature, with more expensive features having more knockoffs.  
A feature is selected only if it beats all of its knockoff counterparts; thus, costlier features have more competition. 
This procedure yields a path of selected feature sets $\R_1, \ldots, \R_p$ for which $\wfdp (\R_k)$ is bounded by a certain computable quantity with high probability, regardless of how $k$ is chosen. 
Unlike existing work on weighted false discovery rate control \citep{benjamini1997multiple,benjamini2007false, ramdas2019unified}, our method provably bounds the weighted false discovery proportion under arbitrary dependence among features.
\citet{yu2021high} recently proposed a predictive modeling method in high-dimensional cost-constrained linear regression problems. Different from their focus which is on good prediction performance under budget constraints, our method aims at recovering the true set of features (as defined in $\H_0^C$) with $\wfdp$ control.

\section{Cheap knockoffs}
\subsection{A review of model-X knockoffs and simultaneous inference} \label{sec:standard}
Our method is based on the model-X knockoff procedure \citep{candes2018panning} and its multiple knockoff extension \citep{2018arXiv181011378R}, which provably control the false discovery rate for arbitrary sample size $n$ and number of features $p$. 
For simplicity, we focus on the following linear model setting 
\begin{align}
  \E \left[Y | X_1, \ldots, X_p \right] = \sum_{j = 1}^p \beta_j X_j,
  \quad \left(X_1, \ldots, X_p \right)^T \sim N(\mathbf{0}, \Sigma).
  \label{eq:model}
\end{align}

We start by briefly reviewing the model-X knockoff approach in the simultaneous inference setting, applied specifically in the linear model \eqref{eq:model}. 
Throughout this paper, we denote $\mathbf{X} \in \real^{n \times p}$ as a data matrix, and $\mathbf{y} \in \real^n$ as a response vector, where $(\mathbf{X}_{i1}, \ldots, \mathbf{X}_{ip}, \mathbf{y}_i) \in \real^{p} \times \real$ are independently and identically distributed as $(X_1, \ldots, X_p, Y)$ for $i = 1, \ldots, n$.
\begin{enumerate}
  \item For each variable $X_j$, construct a knockoff variable $\tilde{X}_j$ that satisfies:
    \begin{enumerate}
      \item $\E(\tilde{X}_j) = \E(X_j)$;
      \item $\Cov(\tilde{X}_j, \tilde{X}_k) = \Cov(X_j, X_k)$ for all $k$;
      \item $\Cov(\tilde{X}_j, X_k) = \Cov(X_j, X_k) - s_j \indi \{j = k\}$ for some $s_j \geq 0$.     
    \end{enumerate}
    The knockoff variables $\tilde{X} = (\tilde{X}_1, \ldots, \tilde{X}_p)$ are constructed to resemble $X$ without any knowledge of the response $Y$. We denote $\tilde{\mathbf{X}} \in \real^{n \times p}$ as the constructed knockoff matrix of $\mathbf{X}$ in a way that $(\tilde{\mathbf{X}}_{i1}, \ldots, \tilde{\mathbf{X}}_{ip})$ is a knockoff of $(\mathbf{X}_{i1}, \ldots, \mathbf{X}_{ip})$ for $i = 1, \ldots, n$.

  \item For each $j \in \{1, \ldots, p\}$, compute statistics $T_j$ and $\tilde{T}_j$ for the variables $X_j$ and $\tilde{X}_j$, respectively. For example, these could be the absolute values of the coefficients of a lasso regression \citep{tibshirani1996regression} on the augmented design matrix $\mathbf{Z} = [\mathbf{X}, \tilde{\mathbf{X}}] \in \real^{n \times 2p}$: 
    \begin{align}
      \hat{\theta}(\lambda) = \argmin_{\theta \in \real^{2p}} \left( \frac{1}{2} \norm{\mathbf{y} - \mathbf{Z} \theta}_2^2 + \lambda \norm{\theta}_1 \right),
      \label{eq:lasso}
    \end{align}
with 
$T_j = |\hat{\theta}(\lambda)_j|$ and $\tilde{T}_j = |\hat{\theta}(\lambda)_{j + p}|$.
    The value of $\lambda$ can be fixed in advance, or selected using cross-validation.
    The knockoff statistics are then defined as
    $W_j = T_j - \tilde{T}_{j}$.
    \citet{barber2015} and \citet{candes2018panning} discuss other choices of $T_j$'s and $W_j$'s.
    Intuitively, a large value of $W_j$ indicates that $X_j$ is a genuine signal variable, i.e., the distribution of $Y$ depends on $X_j$, whereas a small or negative value of $W_j$ indicates that $X_j$ may be irrelevant.

  \item For any ordering of variables $\sigma(1), \ldots, \sigma(p)$, e.g., $|W_{\sigma(1)}| \geq |W_{\sigma(2)}| \geq \ldots \geq |W_{\sigma(p)}|$, report the sets of selected variables $\R_{k} = \left\{ \sigma(j): \sigma(j) \leq \sigma(k), W_{\sigma(j)} > 0 \right\}$, for $k \in \{ 1, \ldots, p\}$.
\end{enumerate}

\citet{katsevich2018towards} work within the simultaneous inference framework \citep{goeman2011multiple}, in which a practitioner wishes to obtain a final set of selected variables with false discovery proportion control when choosing among $\{\R_k, k = 1, \ldots, p\}$.
To allow for such behavior, \citet{katsevich2018towards} form a computable upper bound $\U_k$ such that $\fdp(\R_k)\le \U_k$ holds simultaneously over all $k$ with some known probability.

\subsection{Multiple knockoffs based on cost} \label{sec:mknockoffs}
The knockoff procedure described in the previous section constructs a single knockoff variable for each feature, and then selects features based solely on the values of $W_1, \ldots, W_p$. \citet{barber2015} and \citet{candes2018panning} discuss the possibility of constructing $K$ knockoffs per feature for some value $K>1$ with the goal of achieving higher statistical
power and stability.
This has been pursued in \citet{2018arXiv181011378R} and \citet{emery2019multiple}.

We make a simple yet crucial modification to the multiple knockoffs idea, allowing different features to have different numbers of knockoffs, so that an expensive irrelevant feature will have a lower chance of entering the model than a cheap irrelevant feature.
Assume that the feature costs $\omega_1, \ldots, \omega_p$ are integers with $\omega_j \geq 2$. 
We construct $\omega_j - 1$ knockoff variables for each original variable $X_j$.
If $X_j$ is irrelevant, i.e., $j \in \H_0$, then we expect it to be selected with probability $1/\omega_j$.
We also incorporate costs into the construction of the sequence of selected feature sets $\R_k$.
The cheap knockoff procedure generalizes the multiple knockoff procedure of \citet{2018arXiv181011378R} to the cost-conscious setting:
\begin{enumerate}
  \item For each variable $X_j$ with cost $\omega_j$, denote $\tilde{X}_j^{(1)} = X_j$ and construct the knockoff variables $\tilde{X}_j^{(2)}, \tilde{X}_j^{(3)}, \dots, \tilde{X}_j^{(\omega_j)}$ such that:
    \begin{enumerate}
      \item $\E(\tilde{X}_j^{(\ell)}) = \E(X_j)$ for $\ell \in \{2, \ldots, \omega_j\}$.
      \item $\Cov( \tilde{X}_j^{(\ell)}, \tilde{X}_k^{(m)} ) = \Cov( X_j, X_k ) - s_j \indi\{j = k\} \indi\{ \ell \neq m \}$ for all $\ell \in \{1, \ldots, \omega_j\}$, $m \in \{ 1, \ldots, \omega_k \}$, $j, k \in \{1, \ldots, p \}$, and some constant $s_j \geq 0$.
    \end{enumerate}
    We denote $\tilde{\mathbf{X}}_j^{(\ell)} \in \real^{n}$ as the constructed knockoff variables of $\mathbf{X}_j$, such that $(\tilde{\mathbf{X}}_{ij}^{(\ell)})^{\ell = 1, \ldots, \omega_j}_{j = 1, \ldots, p}$ satisfies the condition above for $(\mathbf{X}_{ij})_{j = 1, \ldots, p}$ for $i = 1, \ldots, n$.

  \item For each $j \in \{1, \ldots, p\}$, compute the statistics $T_j^{(1)}$ (corresponding to the original variable) and $T_j^{(2)}, \dots, T_j^{(\omega_j)}$ (corresponding to the $\omega_j - 1$ knockoff variables).
    For example, these could be the absolute values of the coefficients of the following lasso regression: 
    \begin{align}
      \{\hat{\theta}_{j}^{(\ell)}(\lambda)\}_{j \leq p, \ell \leq \omega_j} = \argmin_{\theta_j^{(\ell)}: j \leq p, \ell \leq \omega_j} \left( \frac{1}{2} \norm{\mathbf{y} - \sum_{j = 1}^p \sum_{\ell = 1}^{\omega_j} \tilde{\mathbf{X}}_j^{(\ell)} \theta_{j}^{(\ell)}}_2^2 + \lambda \sum_{j = 1}^p \sum_{\ell = 1}^{\omega_j} |\theta_j^{(\ell)}| \right),
      \label{eq:lassom}
    \end{align}
    with $T_{j}^{(\ell)} = |\hat{\theta}_j^{(\ell)}(\lambda)|$. The value of $\lambda$ in \eqref{eq:lassom} can be selected using cross-validation.
    We define
    \begin{align}
      \kappa_j = \argmax_{1 \leq \ell \leq \omega_j} T_j^{(\ell)}.
      \label{eq:Wj}
    \end{align}
  \item For any ordering of variables $\sigma(1), \ldots, \sigma(p)$, report the sets of selected variables $\R_{k} = \left\{ \sigma(j): \sigma(j) \leq \sigma(k), \kappa_{\sigma(j)} = 1 \right\}$, for $k \in \{ 1, \ldots, p\}$.
\end{enumerate}

In Step 1, various methods are available for constructing multiple knockoffs given that the distribution of $X$ is known \citep[see, e.g.,][]{candes2018panning, 2018arXiv181011378R}.
The computation of $\kappa_j$ in Step 2 involves the $\omega_j$ statistics $T_j^{(1)}, \ldots, T_j^{(\omega_j)}$; 
$\kappa_j = 1$ indicates that the original variable beats all of its $\omega_j - 1$ knockoff copies.
We show in the supplementary material that the probability of this occurring for an irrelevant feature is inversely proportional to the feature's cost. This is the key property used to show the simultaneous control of the weighted false discovery proportion in the next section.

In principle, any ordering of variables can be used to obtain $\R_k$. In simulations, we consider a specific ordering such that $\tau_{\sigma(1)} \geq \tau_{\sigma(2)} \ldots \geq \tau_{\sigma(p)}$,
where $\tau_j = 2\omega_j^{-1}\{ T^{(\kappa_j)}_j - \max_{\ell \neq \kappa_j} T_j^{(\ell)} \}$.
One reason for this specific choice of $\tau_j$ is that when $\omega_1 = \ldots = \omega_p = 2$,
the above procedure is exactly the same as the standard knockoff procedure reviewed in Section \ref{sec:standard}.
In particular, $W_j > 0$ if and only if $\kappa_j = 1$, and $|W_j| = \tau_j$.
Moreover, all else being equal, we want to make use of cheap features over expensive features.
For this reason, we set $\tau_j$ to be inversely proportional to the feature cost.

\subsection{Simultaneous control of the weighted false discovery proportion} \label{sec:theory}
Having constructed a cost-conscious path of selected variable sets $\R_1, \ldots, \R_p$, we next provide a simultaneous high-probability bound on the weighted false discovery proportion along this path. The next theorem and the remark that follows establish that the computable quantities $\bar{\U}(\R_1, c), \ldots, \bar{\U}(\R_p, c)$, defined below in \eqref{eq:ubar}, simultaneously upper bound $\wfdp (\R_1), \ldots, \wfdp(\R_p)$ with a known probability. 
This means that for any choice of $k$, with high probability our selected feature set is not too wasteful (in terms of the fraction of cost spent on irrelevant features).
\begin{theorem} \label{thm:spotting}
  For any $\alpha \in (0, 1)$, we have 
  \begin{align}
    \Prob \left\{ \wfdp\left(\R_k\right) \leq \U\left(\R_k, c\right)  \text{ \normalfont    for all } k\right\} \geq 1 - \alpha,
    \label{eq:spotting}
  \end{align}
  where for any constant $c > 0$,
  \begin{align}
    \U(\R_k, c) = -\log \alpha \left[\frac{1  + c\sum_{j = 1}^k  \indi\left \{ j \notin \R_k \right \}}{ \left(\sum_{j = 1}^k \omega_j \indi \left\{ j \in \R_k \right\} \right) \vee 1} \right] \left[\max_{k \in \H_0} \frac{\omega_k}{\log \left\{ \omega_k - \left( \omega_k - 1 \right) \alpha^{c} \right\}}\right].
    \label{eq:wfdpbar}
  \end{align}
\end{theorem}
For the standard knockoff procedure described in Section \ref{sec:standard}, we have $\omega_1 = \ldots = \omega_p = 2$. In that case, with $c = 1$, \eqref{eq:wfdpbar} reduces exactly to the bound from applying Theorem 2 of \citet{katsevich2018towards} to the Selective and Adaptive SeqStep procedure \citep{barber2015} with $p_\ast = \lambda = 1/2$.

As mentioned in Section \ref{sec:standard}, our procedure can be generalized to any known distribution of $X$ and any unknown conditional distribution of $Y$ given $X$. For example, in the binary classification data example in Section \ref{sec:data}, we consider the statistics $\{T_j^{(\ell)}\}$ derived from $\ell_1$-penalized logistic regression. Following the arguments \citet{candes2018panning}, we can show that Theorem \ref{thm:spotting} also holds for this choice of $\{T_j^{(\ell)}\}$.

\begin{remark}
  The weighted false discovery proportion upper bound $\U(\R_k, c)$ depends on the unknown set $\H_0$. In practice, we can use an upper bound
  \begin{align}
    \bar{\U}(\R_k, c) = -\log \alpha \left[\frac{1  + c\sum_{j = 1}^k  \indi\left \{ j \notin \R_k \right \}}{\left(\sum_{j = 1}^k \omega_j \indi \left\{ j \in \R_k \right\} \right) \vee 1} \right] \left[\max_{k} \frac{\omega_k}{\log \left\{ \omega_k - \left( \omega_k - 1 \right) \alpha^{c} \right\}}\right].
    \label{eq:ubar}
  \end{align}
  Moreover, if an estimated set $\hat{\H}_0$ satisfying $\H_0 \subseteq \hat{\H}_0$ is available, then \eqref{eq:wfdpbar} with the maximum taken over $\hat{\H}_0$ gives a tighter bound in \eqref{eq:spotting}.
\end{remark}

Our procedure yields a sequence of sets $\R_k$ of selected variables, and the bound in \eqref{eq:spotting} gives a specific description of the tradeoff between capturing enough of the signal variables and incurring too much cost. The simultaneous nature of the bound means that $\wfdp(\R_k)$ is controlled regardless of the approach used to select $k$: the choice of $k$ can depend on the size of $\R_k$, the cost of $\R_k$, or in fact any function of the data.

\section{Simulation studies} \label{sec:simulation}
We now investigate the feature selection performance of cheap knockoffs in simulation. 
We set $n = 200$ and $p = 30$. 
Each element of the design matrix $\mathbf{X} \in \real^{n \times p}$ is independent and identically distributed as $N(0, 1)$. 
The response is generated from the linear model \eqref{eq:model} with Gaussian errors $\varepsilon \sim N(0, \sigma^2)$ and $\sigma^2=(4n)^{-1}\snorm{\mathbf{X} \beta}_2^2$.
We let $\beta_{1} = \ldots = \beta_{10} = 2$, and $\beta_j = 0$ for $j > 10$.
We set the first half of the relevant features to be expensive and the second half to be cheap, i.e., $\omega_{1} = \ldots = \omega_5 = 6$, and $\omega_{6} = \ldots = \omega_{10} = 2$. 
For the irrelevant features, i.e., for any $j > 10$, 
we set $\Prob(\omega_j = 6) = \gamma$ and $\Prob(\omega_j = 2) = 1 - \gamma$, where $\gamma \in \{0, 0.25, 0.5, 0.75, 1\}$.

We construct multiple knockoff variables using entropy maximization \citep{2018arXiv181011378R}, and we compute the statistics $T_j^{(\ell)}$ as the absolute value of the lasso coefficient estimates in \eqref{eq:lassom},
with the tuning parameter selected using cross-validation.
In Appendix \ref{app:time} we report the wall-clock running time of cheap knockoffs in the numerical studies.
We find that the majority of computation is spent on generating multiple knockoffs, which is challenging when $p$ is large and (or) the feature costs are large (after dividing by their greatest common factor).
In such cases, alternative construction methods could be used. For example, \citet[Appendix A.1.][]{2018arXiv181011378R} show that an equicorrelation construction has a closed form expression, which is particularly favorable in computation since it does not depend on the number of multiple knockoffs (and equivalently, the feature costs). 

We first verify the bound in Theorem \ref{thm:spotting} and compare the performance of cheap knockoffs to \citet{katsevich2018towards}, which ignores feature costs.
In particular, by carrying out Steps 1-3 in Section \ref{sec:standard} with $\omega_1 = \ldots = \omega_p = 2$ in \eqref{eq:ubar}, the bound in \eqref{eq:ubar} coincides with the result in \citet{katsevich2018towards}. We denote this approach as \citet{katsevich2018towards}. For both methods, we take $\alpha = 0.2$ in \eqref{eq:ubar}.
In Fig.~\ref{fig:ratio} we report both the ratio $\bar{\U}(\R_k, 1)^{-1}\wfdp(\R_k)$ and the actual weighted false discovery proportion $\wfdp(\R_k)$ for each $\R_k$ for both methods in the settings where $\gamma = 0, 0.5,$ and $1$.
\begin{figure}[H]
  \centering
  \includegraphics[width=0.9\textwidth]{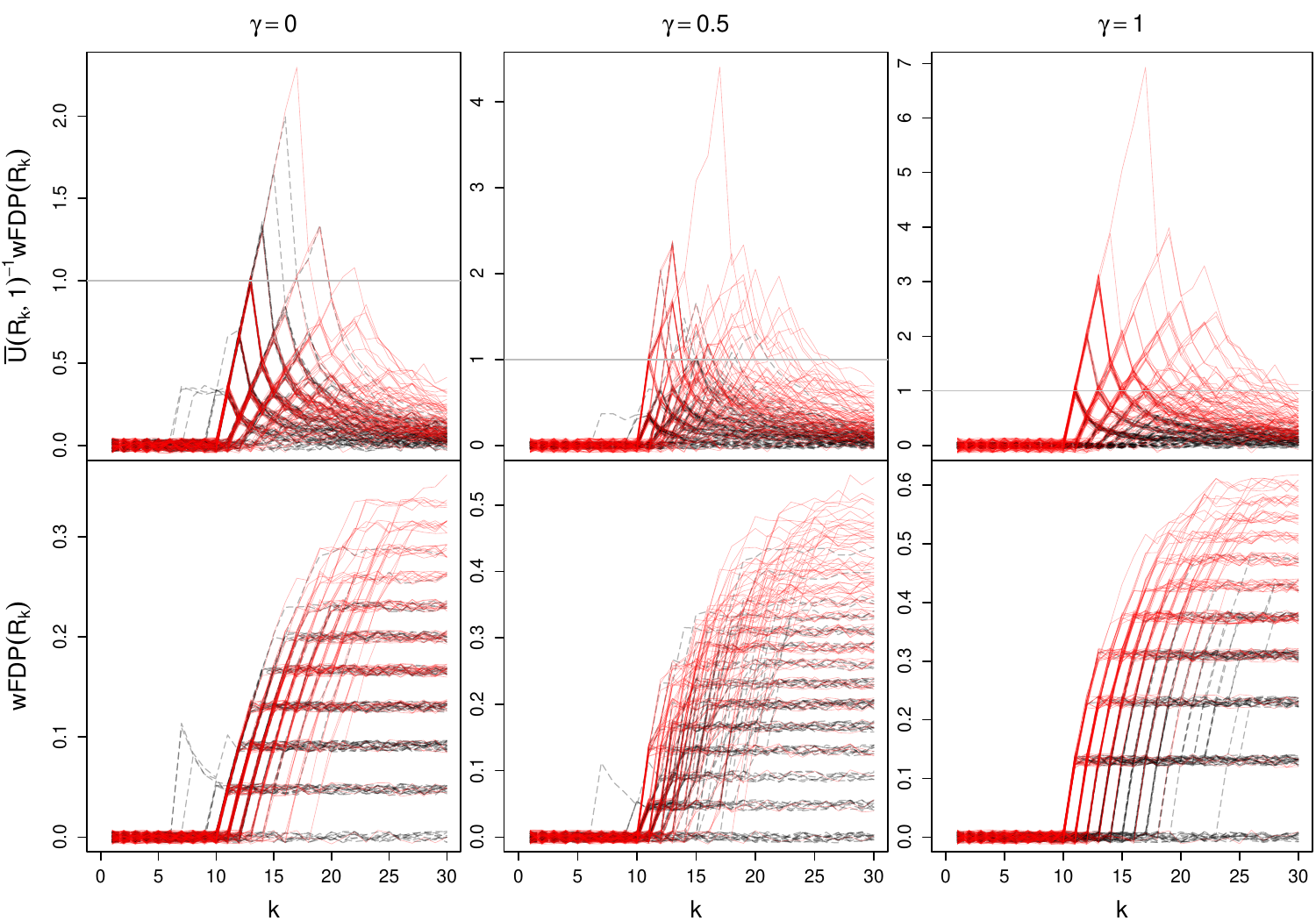}
  \caption{Each line represents one of 100 simulated datasets. Jitter is applied to ease visualization. The black dashed lines represent cheap knockoffs (our proposal) which incorporates feature costs, and the red solid lines represent \citet{katsevich2018towards} which does not make use of feature costs. Top panel: the cheap knockoffs approach controls the weighted false discovery proportion with the desired probability ($\alpha=0.2$) whereas the \citet{katsevich2018towards} procedure does not. Bottom panel: The cheap knockoffs attains lower weighted false discovery proportion than the \citet{katsevich2018towards} procedure for most values of $k$ when $\gamma$ is large.}
  \label{fig:ratio}
\end{figure}
As seen in Fig.~\ref{fig:ratio}, the ratio $\bar{\U}(\R_k, 1)^{-1}\wfdp(\R_k)$ for our cheap knockoff procedure is mostly below $1$, indicating that the bound in Theorem \ref{thm:spotting} holds.
Moreover, when $\gamma$ is large, the weighted false discovery proportion for the cheap knockoff procedure is lower than \citet{katsevich2018towards} for most values of $k$.
Table~\ref{tab:ratio} gives the estimated probability that the bound is violated, i.e., $\widehat{\Prob}(\sup_{k} \bar{\U}_k^{-1}(\R_k, 1) \wfdp(\R_k) > 1)$, for each method for $\gamma \in \{0, 0.25, 0.5, 0.75, 1\}$.
\begin{table}[H]
  \begin{center}
    \begin{tabular}{ c | c | c | c | c | c}
      \hline
      $\gamma$ & 0 & 0.25 & 0.5 & 0.75 & 1\\ \hline
      Cheap knockoffs (our proposal) & 0.08 & 0.05 & 0.08 & 0.07 & 0.04\\ \hline
      \citet{katsevich2018towards}  & 0.01 & 0.05 & 0.12 & 0.25 & 0.31\\
      \hline
    \end{tabular}
  \end{center}
  \caption{Proportion of 100 simulated datasets for which $\sup_{k} \bar{\U}_k^{-1}(\R_k, 1) \wfdp(\R_k) > 1$ is violated. Our proposed cost-conscious procedure successfully controls the probability below the $\alpha = 0.2$ level for all values of $\gamma$, while \citet{katsevich2018towards} does not control this probability when $\gamma = 0.75$ and $\gamma = 1$.}
  \label{tab:ratio}
\end{table}

We see that the \citet{katsevich2018towards} procedure which is not cost-conscious performs worse as $\gamma$ increases, that is, when irrelevant variables are more likely to be expensive. Since the method ignores cost, it may erroneously select expensive irrelevant features, leading to poor weighted false discovery proportion.

While our proposal focuses on recovering the correct set of features with simultaneous $\wfdp$ control, we show empirically that the set of features selected by cheap knockoffs usually incurs low cost without compromising prediction accuracy.
Specifically, for each set of selected variables $\R_1, \ldots, \R_p$, we compute both the root mean squared prediction error of the least squares model fit to the variables in $\R_k$, and the total cost $\sum_{j \in \R_k} \omega_j$. 
We see from Fig.~\ref{fig:pred} that for a given budget, the cheap knockoff procedure attains smaller prediction error than the procedure in \citet{katsevich2018towards}, which is not cost-conscious. In particular, the cheap knockoff procedure tends to select all five of the cheap relevant features before any expensive feature is let in the model, whereas \citet{katsevich2018towards} does not take feature cost into consideration.
For $k \geq 10$, $\R_k$ for both methods includes essentially all the relevant features, thus giving similar performance.
\begin{figure}[H]
  \centering
  \includegraphics[width=0.9\textwidth]{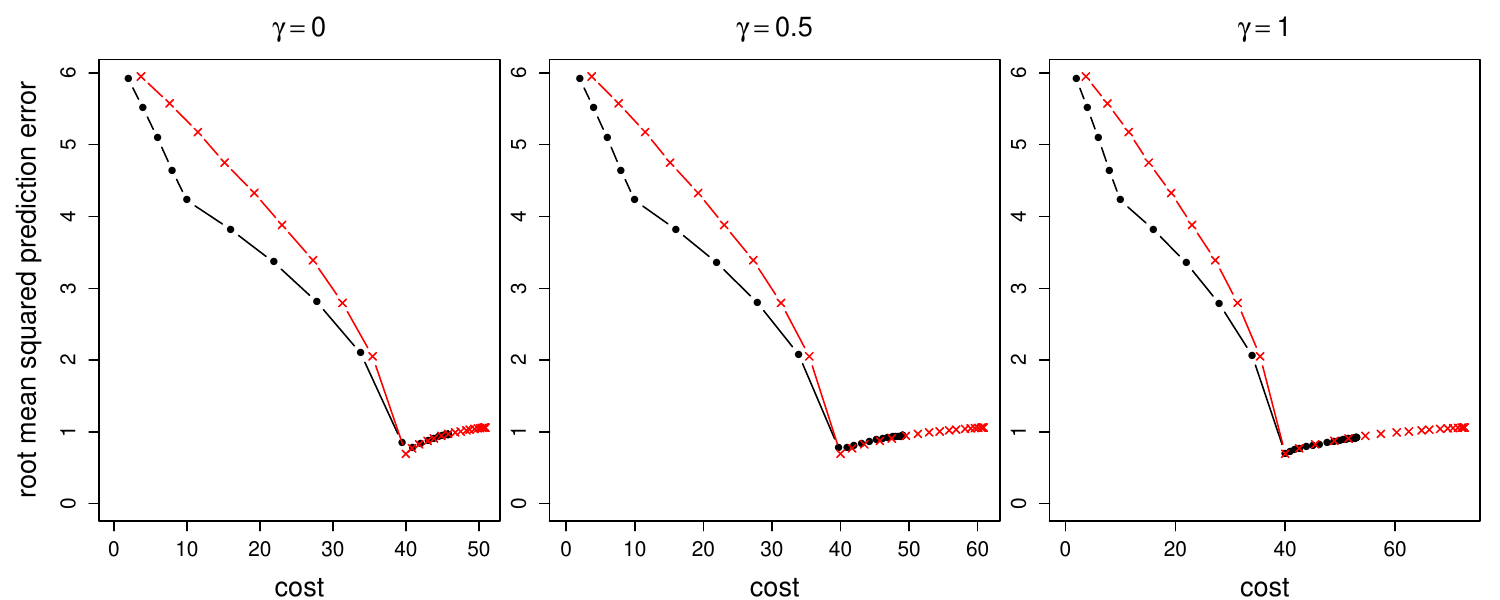}
  \caption{Tradeoff between prediction accuracy and total cost (averaged over 100 simulations). The line with dots in black represents the cheap knockoff procedure, and the line with crosses in red represents \citet{katsevich2018towards}. The cost of the model selected by our cost-conscious procedure can be much lower than that of the procedure in \citet{katsevich2018towards} without sacrificing predictive performance.}
  \label{fig:pred}
\end{figure}

\section{Data application} \label{sec:data}
To gauge the performance of cheap knockoffs in a real dataset, we consider data from the National Health and Nutrition Examination Survey (NHANES) (\citealt{nhanes},  processed in \citealt{kachuee2019opportunistic, kachuee2019cost}). 
The dataset contains 92062 samples of survey participants.
We consider 30 features, which can be broadly categorized into four types: 
demographics, questionnaire-based, examination-based, and laboratory-based.
For each feature, medical experts suggest a corresponding integer-valued cost (ranging from 2 to 9) for that feature based on ``the overall financial burden, patient privacy, and patient inconvenience'' \citep{kachuee2019cost}. A brief summary of the 30 features can be found in Table~\ref{tab:features}. Finally, each observation is associated with a label of pre-diabetes/diabetes (as one category) or normal. 
The task is to select features that are closely associated with diabetes while taking feature cost into consideration.

\begin{table}[H]
  \begin{tabular}{lll}
                      & \textbf{Examples}                                                               & \textbf{Cost}    \\
    \textbf{Demographics}           & Age; Income; Education level                                   & 2 to 4 \\
    \textbf{Questionnaire}          & Average sleep length (in hours) & 4        \\
    \textbf{Examination}            & Diastolic Blood pressure; Systolic Blood Pressure     & 5      \\
    \textbf{Laboratory}             & Cholesterol; Triglyceride; Fibrinogen                 & 9     
  \end{tabular}
  \caption{Examples of the features in the NHANES dataset}
  \label{tab:features}
\end{table}

We consider the cheap knockoff procedure as in Section \ref{sec:mknockoffs}, modified so that 
the statistics $\{T_j^{(\ell)}\}$ computed in \eqref{eq:lassom} are derived from $\ell_1$-penalized logistic regression (instead of $\ell_1$-penalized least squares). Following the arguments in \citet{candes2018panning}, we can show that Theorem \ref{thm:spotting} also holds for this choice of $\{T_j^{(\ell)}\}$.

To numerically verify Theorem \ref{thm:spotting}, we would need to know the true set of relevant variables.
We test the cheap knockoff procedure using partially-simulated data. 
To form a reasonable ground truth, we start by performing logistic regression on a random set of 72062 samples. In total, we retain 11 variables whose $p$-values are smaller than 0.01 / 30 (by Bonferroni correction).
We take these as the true set of relevant variables (see Appendix \ref{app:truth} for the list of relevant variables).
We next generate responses for the remaining 20000 samples from a logistic regression model using only these selected features. The coefficient values used correspond to those from the fitted logistic regression estimates.
We then randomly divide these 20000 samples (with simulated responses) into 50 non-overlapping sets, each containing 400 samples. On each set, we run our method to obtain a path of selected variables. Finally, we compute the estimated probability that the bound in \eqref{eq:wfdpbar} is violated, i.e., $\widehat{\Prob}(\sup_{k} \bar{\U}_k^{-1}(\R_k, 1) \wfdp(\R_k) > 1)$ for $\alpha \in \{0.05, 0.1, ..., 0.5\}$. 
We see from Table \ref{tab:ratio_data} that the estimated probability is lower than the corresponding value of $\alpha$, indicating that Theorem \ref{thm:spotting} holds for our proposed cost-conscious procedure.
\begin{table}[H]
  \begin{center}
    \begin{tabular}{ c | c | c | c | c | c | c | c | c | c | c}
      \hline
      $\alpha$ & 0.05 & 0.10 & 0.15 & 0.20 & 0.25 & 0.30 & 0.35 & 0.40 &0.45 &0.50\\ \hline
      Cheap knockoffs & 0.04 & 0.04 & 0.04 & 0.04 & 0.04 &  0.04  & 0.04   &  0.04   & 0.04 & 0.06 \\
      \hline
    \end{tabular}
  \end{center}
  \caption{Proportion of 50 data subsets for which $\sup_{k} \bar{\U}_k^{-1}(\R_k, 1) \wfdp(\R_k) > 1$ is violated. }
  \label{tab:ratio_data}
\end{table}

On each of the 50 non-overlapping data subsets, we further compute $\wfdp$ and cost for the path of selected variables $\R_k$ returned by 
cheap knockoffs and the proposal in \citet{katsevich2018towards}, which ignores feature costs.
Figure~\ref{fig:NHANES_wfdp} reports the $20$, $50$, and $80$ percentiles (over the 50 non-overlapping sets) of $\wfdp$ and cost, and shows that
our proposal effectively attains a lower $\wfdp$ and a lower cost than the proposal in \citet{katsevich2018towards} which is not cost-conscious.
\begin{figure}[H]
  \centering
  \includegraphics[width=.9\textwidth]{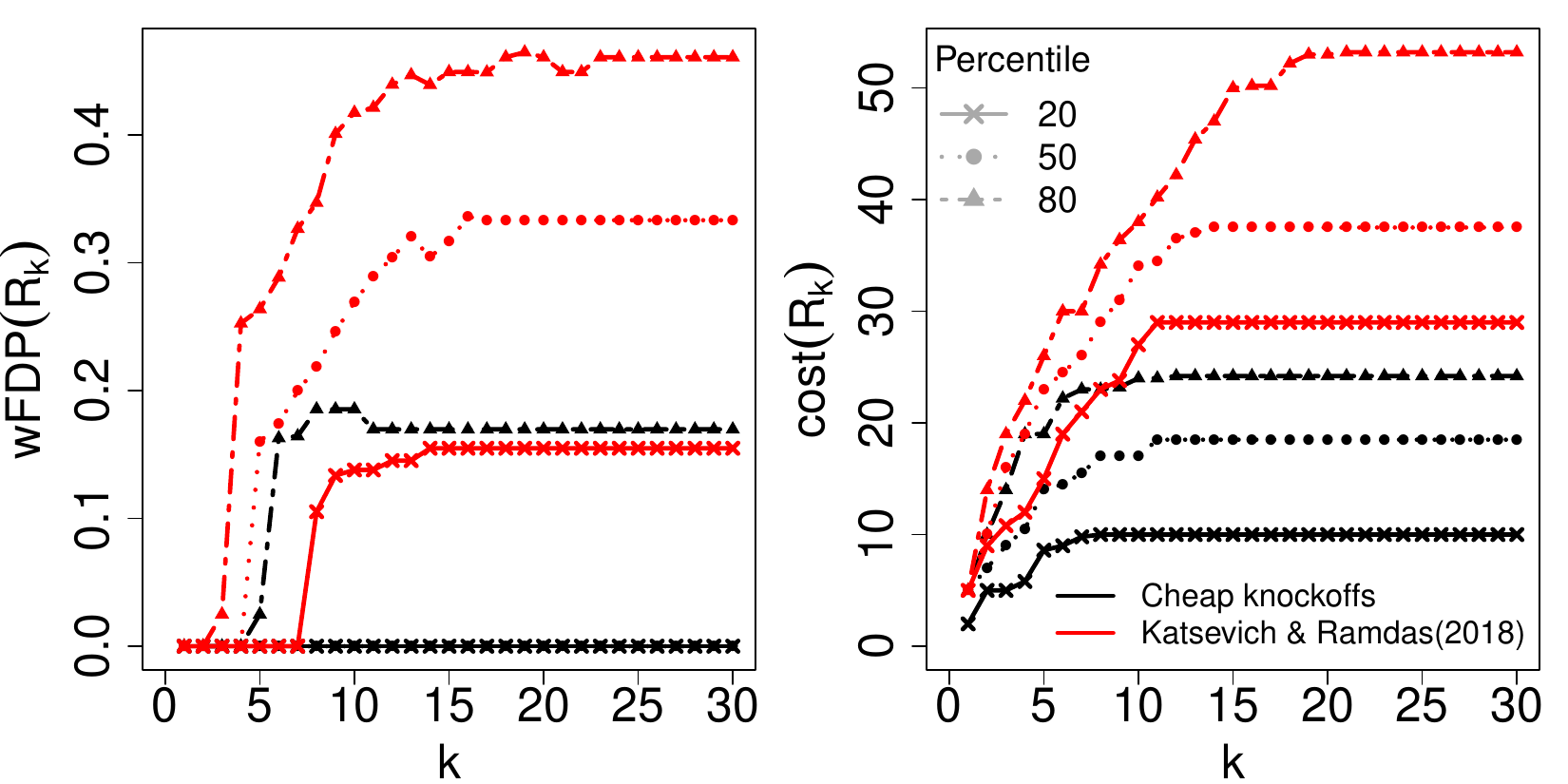}
  \caption{The 20, 50, and 80 percentiles of $\wfdp$ (left panel) and cost (right panel) over 50 non-overlapping data subsets of cheap knockoffs and the procedure in \citet{katsevich2018towards}. \\}
  \label{fig:NHANES_wfdp}
\end{figure}

Although prediction performance of the selected model is not the main theoretical focus of our proposal, we next study the prediction performance and the total cost of the selected variables. For comparison, we consider the following methods:
\begin{enumerate}
  \item \textbf{Katsevich \& Ramdas(2018)}: the proposal of \citet{katsevich2018towards} applied to the `Selective and adaptive SeqStep' method. It is equivalent to our method if we ignore the cost information, i.e., we set $\omega_1 = \omega_2 = ... = \omega_{30} = 2$.
  \item \textbf{Logistic regression}: logistic regression applied to all 30 features. This procedure is not cost-conscious, and does not perform features selection. We use this as a benchmark for classification performance.
\end{enumerate}
We run these methods on all 92062 observations.
Given the large sample size, we expect training error to be a good approximation of the generalization error. 
Furthermore, to highlight the effects of feature costs, we consider exaggerating the feature costs by using the squares of their actual costs.
\begin{figure}[H]
  \centering
  \includegraphics[width=\textwidth]{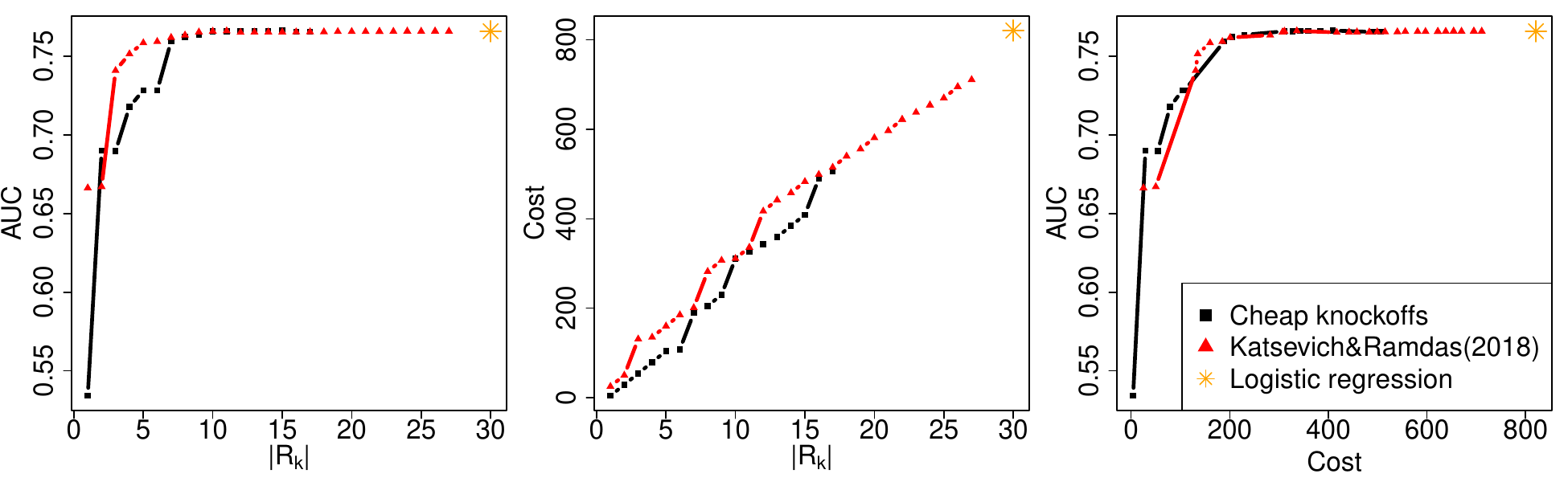}
  \caption{\emph{Left}: The classification performance (in terms of the area under the ROC curve) for different sizes of the selected model $\mathcal{R}_k$ ($k = 1, \ldots, 30$). \emph{Center}: The total cost for different sizes of the selected model.
  \emph{Right}: The classification performance versus the cost of the selected model.
In all three panels of this figure, we consider the squared costs to highlight the effects of feature costs.}
  \label{fig:NHANES}
\end{figure}
From Figure~\ref{fig:NHANES}, we see that cheap knockoffs can achieve favorable classification performance at a low feature cost.
In particular, the first two panels show that for a fixed model size, cheap knockoffs tends to achieve slightly worse classification performance than the procedure of \citet{katsevich2018towards}, which is not cost-conscious. 
However, our method achieves this classification performance at a lower cost.
The right panel shows that for a given model cost, our method can obtain favorable classification performance compared with the proposal of \citet{katsevich2018towards}.
Moreover, our method's classification performance is close to the benchmark of logistic regression, while using a much cheaper set of features. 

In Figures~\ref{fig:path} and \ref{fig:path_exp}, 
we show the path of variables selected by cheap knockoffs and that of \citet{katsevich2018towards}.
Each point represents a variable added to a model (with the feature name in the legend).
For example, we see that both methods include \texttt{Gender}, \texttt{Height}, \texttt{Weight}, and \texttt{Triglyceride} when the model size is 4. 
However, the cheap knockoff procedure tends to select cheaper features first, adding the expensive laboratory feature \texttt{Triglyceride} last among these four features.
By comparison, the proposal of \citet{katsevich2018towards} does not show any preference for inexpensive features.
For the model with two variables, cheap knockoffs selects \texttt{Gender} and \texttt{Height}, which has lower cost and better classification performance than the model of \texttt{Height} and \texttt{Weight} selected by \citet{katsevich2018towards}.

In addition, in Figure~\ref{fig:path_exp}, we present the path of variables selected by cheap knockoffs applied with squared feature costs, where squaring has been performed to exaggerate the effect of the feature costs. 
Comparing with Figure~\ref{fig:path}, we see that cheap knockoffs tends to select less expensive features, while still attaining comparable classification performance. In particular, when the costs are squared, cheap knockoffs no longer selects \texttt{Diastolic BP(2nd)}, \texttt{Systolic BP(4th)}, \texttt{Systolic BP(1st)}, \texttt{Diastolic BP(3rd)}, \texttt{Vigorous activity}, and \texttt{Upper leg length}. Among these omitted variables, only \texttt{Upper leg length} is considered relevant by the logistic regression (see Appendix \ref{app:truth}).

\begin{figure}
  \centering
  \includegraphics[width=0.7\linewidth]{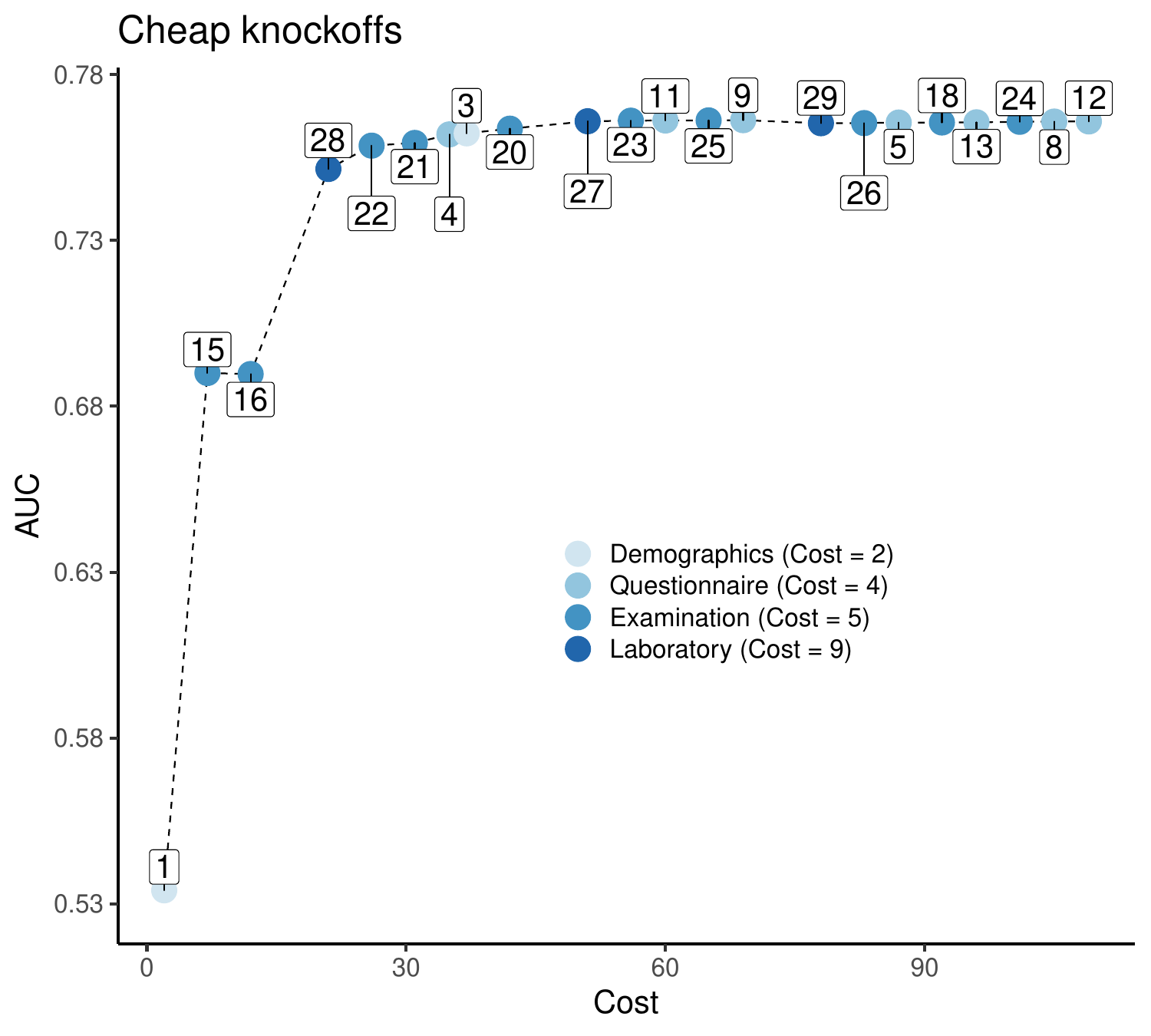}
  \includegraphics[width=.7\linewidth]{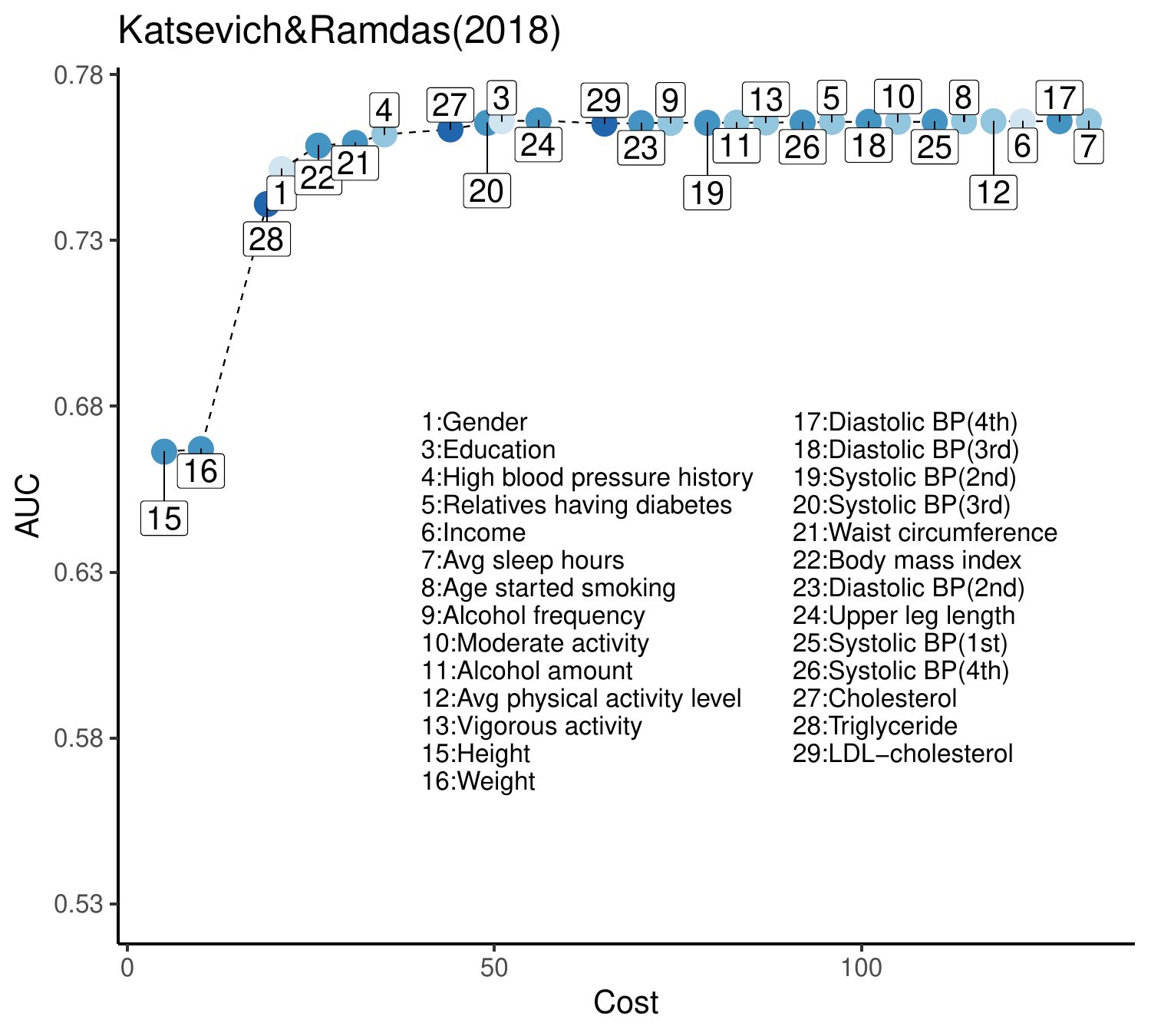}
  \caption{The path of variables selected by cheap knockoffs (top) and the proposal of \citet{katsevich2018towards} (bottom). Each point represents a newly selected feature in the model. Variable indices are ordered from cheapest to most expensive.}
  \label{fig:path}
\end{figure}
\begin{figure}
  \centering
  \includegraphics[width=.7\linewidth]{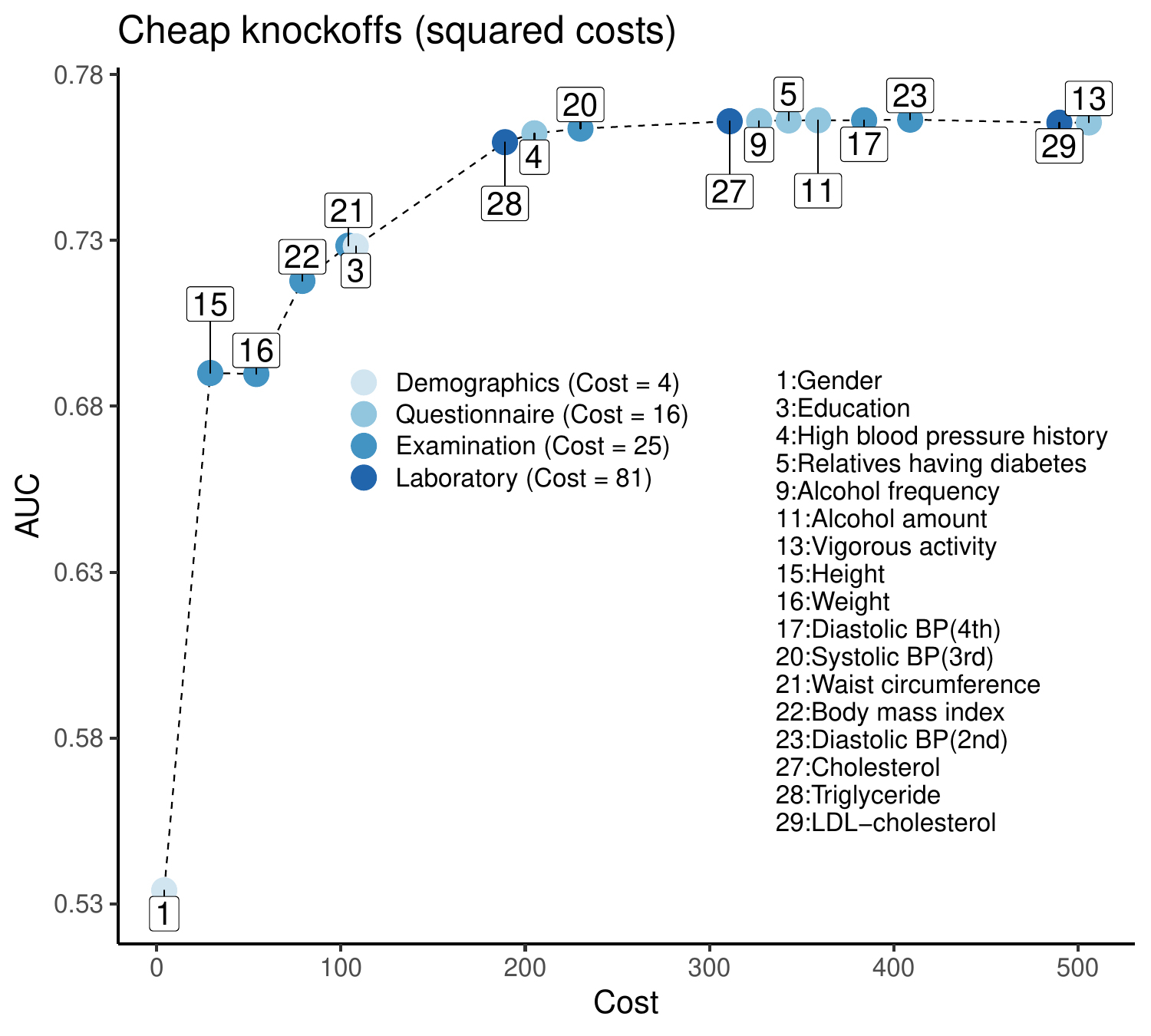}
  \caption{The path of variables selected by cheap knockoffs, with squared costs. Each point represents a newly selected feature in the model. Variable indices are ordered from cheapest to most expensive.}
  \label{fig:path_exp}
\end{figure}

\section{Discussion}
In this paper, we proposed cheap knockoffs, a procedure for performing feature selection when features have costs. Cheap knockoffs is based on the idea of constructing multiple knockoffs for each feature.
In particular, cheap knockoffs forces more expensive features to compete with more knockoffs, making it harder for expensive features to be selected.
Our method yields a path of selected feature sets, and we show that the weighted false discovery proportion is simultaneously bounded with high probability along this path.

An interesting yet challenging future research direction is to develop a method based on the multiple knockoffs idea that provably controls the weighted false discovery rate. The martingale-type arguments used in the original knockoff paper rely on certain symmetries that are broken when the numbers of knockoffs constructed for different features are not all equal.

Finally, an \texttt{R} package named \texttt{cheapknockoff}, implementing our proposed method, is available on \url{https://github.com/hugogogo/cheapknockoff}. 
The simulation studies in Section \ref{sec:simulation} use the \texttt{simulator} package \citep{2016arXiv160700021B}, and the code to reproduce the simulation results (in Section \ref{sec:simulation}) and the NHANES data analysis (in Section \ref{sec:data}) is available at \url{https://github.com/hugogogo/reproducible/tree/master/cheapknockoff}. 
The NHANES dataset \citep{nhanes} is processed in \citet{kachuee2019opportunistic, kachuee2019cost} and is available at \url{https://github.com/mkachuee/Opportunistic}. 

\section{Acknowledgments}
\label{sec:acknowledgments}
The authors thank Will Fithian for suggesting the simultaneous inference framework.
All authors were supported by NIH Grant R01GM123993. Jacob Bien was also supported by NSF CAREER
Award DMS-1653017 and Daniela Witten was also partially supported by NIH Grant DP5OD009145,
NSF CAREER Award DMS-1252624, and a Simons Investigator Award in Mathematical Modeling of Living Systems.
\bibliographystyle{plainnat}
\bibliography{short.bib}

\newpage
\appendix
\appendixpage
\section{NHANES dataset: significant features in logistic regression} \label{app:truth}
In the order of increasing $p$-values (smaller than 0.01 / 30):
\begin{table}[H]
  \centering
  \begin{tabular}{ll}
    \hline
    Name & $p$-value \\
    \hline
    Gender & $1.73 \times 10^{-262}$ \\
    Triglyceride & $5.92 \times 10^{-214}$ \\
    Height & $1.17 \times 10^{-184}$ \\
    Weight & $1.98 \times 10^{-102}$ \\
    Waist circumference & $4.09 \times 10^{-37}$ \\
    Body mass index & $4.02 \times 10^{-31}$ \\
    High blood pressure history & $1.51 \times 10^{-27}$ \\
    Cholesterol & $4.92 \times 10^{-24}$ \\
    Education & $8.16 \times 10^{-10}$ \\
    Upper leg length  & $3.17 \times 10^{-5}$ \\
    Systolic BP(3rd)  & $1.01 \times 10^{-4}$\\
    \hline
  \end{tabular}
\end{table}

\section{Running time comparison in numerical studies} \label{app:time}
\begin{table}[H]
  \centering
  \begin{tabular}{ c | c | c | c | c | c}
    \hline
    $\gamma$ & 0 & 0.25 & 0.5 & 0.75 & 1\\ \hline
    Cheap knockoffs (our proposal) & 2.796 & 2.772 & 2.784 & 2.798 & 2.812\\ \hline
    \citet{katsevich2018towards}  & 0.273 & 0.250 & 0.258 & 0.251 & 0.253\\
    \hline
  \end{tabular}
  \caption{Wall-clock time comparison (in seconds, averaged over 100 simulated datasets) between our proposal and \citet{katsevich2018towards} in generating Table~\ref{tab:ratio}. }
  \label{tab:time_simu}
\end{table}
\begin{table}[H]
  \centering
  \captionsetup{type=table} 
  \begin{tabular}{ c | c }
    \hline
    Cheap knockoffs (our proposal) & 7.284 \\ \hline
    \citet{katsevich2018towards}  & 2.678 \\
    \hline
  \end{tabular}
  \caption{Wall-clock time comparison (in seconds, averaged over 50 non-overlapping data subsets) between our proposal and \citet{katsevich2018towards} in generating Figure~\ref{fig:NHANES_wfdp}. }
  \label{tab:time_data}
\end{table}

\section{Properties of multiple knockoffs} \label{app:mkv}
We study the properties of the multiple knockoffs constructed in Step 1 of Section \ref{sec:mknockoffs}. 
Define 
\begin{align}
  \tilde{Z} = \left(\tilde{X}_1^{(2)}, \ldots, \tilde{X}_1^{(\omega_1)}, \tilde{X}_2^{(2)}, \ldots, \tilde{X}_2^{(\omega_2)}, \ldots, \tilde{X}_p^{(2)}, \ldots, \tilde{X}_p^{(\omega_p)}\right)^T \in \real^{\sum_j(\omega_j - 1)}
  \nonumber
\end{align}
as the random vector of all knockoff features, and 
\begin{align}
  Z = \left(\tilde{X}_1^{(1)}, \tilde{X}_1^{(2)}, \ldots, \tilde{X}_1^{(\omega_1)}, \tilde{X}_2^{(1)}, \tilde{X}_2^{(2)}, \ldots, \tilde{X}_2^{(\omega_2)}, \ldots, \tilde{X}_p^{(1)}, \tilde{X}_p^{(2)}, \ldots, \tilde{X}_p^{(\omega_p)}\right)^T \in \real^{\sum_j\omega_j},
  \label{eq:Z}
\end{align}
where $\tilde{X}_j^{(1)} = X_j$ is the original feature for $j = 1, \ldots, p$. 
For any $p$-tuple of permutations $\varsigma = (\varsigma_1, \ldots, \varsigma_p)$ where $\varsigma_j$ is a permutation on the set $\{1, \ldots, \omega_j\}$, and for any vector $v = (v_1^{(1)}, \ldots, v_1^{(\omega_1)}, \ldots, v_p^{(1)}, \ldots, v_p^{(\omega_p)}) \in \real^{\sum_j \omega_j}$, we define
\begin{align}
  v_{\swap(\varsigma)} = \left( v_1^{(\varsigma_1(1))}, \ldots, v_1^{(\varsigma_1(\omega_1))}, v_2^{(\varsigma_2(1))}, \ldots, v_2^{(\varsigma_2(\omega_2))}, \ldots, v_p^{(\varsigma_p(1))}, \ldots, v_p^{(\varsigma_p(\omega_p))} \right)^T \in \real^{\sum_j \omega_j}.
  \nonumber
\end{align}
Therefore, $Z_{\swap(\varsigma)}$ denotes the random vector where each $\varsigma_j$ permutes the $\omega_j$ knockoff features (including the original one) corresponding to $X_j$.

We generalize the definition of multiple model-X knockoffs \citep[Definition 3.2 in][]{2018arXiv181011378R} to our setting in which each feature can have a different number of knockoffs:
\begin{definition}
  Consider any cost vector $\omega = (\omega_1, \ldots, \omega_p)$, where $\omega_j > 1$ are integers. The random vector $\tilde{Z}$ is a valid $\omega$-knockoff of $X = (X_1, \ldots, X_p)$ if 
  \begin{enumerate}
    \item $Z_{\swap(\varsigma)}$ and $Z$ are identically distributed for any tuple of permutations $\varsigma = (\varsigma_1, \ldots, \varsigma_p)$;
    \item $\tilde{Z}$ and $Y$ are conditionally independent given $X$. 
  \end{enumerate}
\end{definition}
Under the assumption that $X$ follows a multivariate Gaussian distribution, it can be verified \citep[see, e.g., Proposition 3.4 in][]{2018arXiv181011378R} that following Step 1 in Section \ref{sec:mknockoffs}, the vector $\tilde{Z}$ is a valid $\omega$-knockoff of $X$. In particular, the second property is guaranteed provided that the construction of $\tilde{Z}$ does not use $Y$, as in \citet{2018arXiv181011378R}.

The next lemma states the exchangeability property of the irrelevant features and their knockoffs, i.e., we can permute an irrelevant feature and its knockoffs without changing the joint distribution of $Z$ and $Y$.
\begin{lemma}[Exchangeability of irrelevant features and their knockoffs] \label{lem:exchange} Consider any tuple of permutations $\varsigma = (\varsigma_1, \ldots, \varsigma_p)$, where $\varsigma_j$ is the identity permutation for $j \notin \H_0$, and $\varsigma_j$ is an arbitrary permutation over the set $\{1, \ldots, \omega_j\}$ for $j \in \H_0$. If $\tilde{Z}$ is a valid $\omega$-knockoff of $X$, then $(Z, Y)$ and $(Z_{\swap(\varsigma)}, Y)$ are identically distributed.
\end{lemma}
\begin{proof}
  By the property of a valid $\omega$-knockoff, $Z_{\swap(\varsigma)}$ and $Z$ are identically distributed. So it is left to show that $Y | Z$ and $Y | Z_{\swap(\varsigma)}$ are identically distributed. This can be shown using the same arguments as in the proof of Lemma 1 in \citet{candes2018panning}.
\end{proof}

We denote
\begin{align}
  T = \left( T_1^{(1)}, \ldots, T_1^{(\omega_1)}, T_2^{(1)}, \ldots, T_2^{(\omega_2)}, \ldots, T_p^{(1)}, \ldots, T_p^{(\omega_p)} \right) \in \real^{\sum_j \omega_j},
  \nonumber
\end{align}
for $T_j^{(\ell)}$ defined in Step 2 of Section \ref{sec:mknockoffs}.
Furthermore, we define component-wise order statistics on $T$,
\begin{align}
  T_{\mathrm{ordered}} =\left( T_{1, (1)}, \ldots, T_{1, (\omega_1)}, T_{2, (1)}, \ldots, T_{2, (\omega_2)}, \ldots, T_{p, (1)}, \ldots, T_{p, (\omega_p)} \right) \in \real^{\sum_j \omega_j}
  \nonumber
\end{align}
such that $T_{j, (1)} \geq T_{j, (2)} \geq \ldots \geq T_{j, (\omega_j)}$ for all $j$. 

The following lemma characterizes the multiple knockoff statistics $\{\kappa_j\}_{j = 1}^p$ computed in Step 2 of Section \ref{sec:mknockoffs}.
It essentially states that for $j \in \H_0$, the statistics $\kappa_j$ corresponding to the irrelevant feature $X_j$ is uniformly distributed on the set $\{1, \ldots, \omega_j\}$, and is independent of the statistics corresponding to all other features and the component-wise order statistics $T_{\mathrm{ordered}}$.
This property generalizes the ``coin-flip'' property of the standard model-X knockoff \citep[see, e.g., Lemma 2 in][]{candes2018panning}, and is the key to the proof of Theorem \ref{thm:spotting}.
\begin{lemma}[Multiple knockoff statistics] \label{lem:uniform}
  Suppose $\tilde{Z}$ is a valid $\omega$-knockoff of $Z$. For any $j \in \H_0$, the statistic $\kappa_j$ is uniformly distributed on the set $\{1, \ldots, \omega_j\}$, and is independent of $\{\kappa_k\}_{k \neq j}$ and the order statistics $T_{\mathrm{ordered}}$.
\end{lemma}
\begin{proof}
  We adapt the proof idea in B.2 of \citet{2018arXiv181011378R}. Consider any tuple of permutations $\varsigma = (\varsigma_1, \ldots, \varsigma_p)$, where $\varsigma_j$ is the identity permutation for $j \notin \H_0$, and $\varsigma_j$ is an arbitrary permutation over the set $\{1, \ldots, \omega_j\}$ for $j \in \H_0$. We first show that $(\varsigma_1(\kappa_1), \ldots, \varsigma_p(\kappa_p), T_{\mathrm{ordered}})$ has the same distribution as $(\kappa_1, \ldots, \kappa_p, T_{\mathrm{ordered}})$.

  We denote $\varsigma^{-1} = (\varsigma_1^{-1}, \ldots, \varsigma_p^{-1})$ where $\varsigma_j^{-1}$ is the inverse permutation of $\varsigma_j$.
  Recall from Step 2 of Section \ref{sec:mknockoffs}, combined with the definition of $Z$ in \eqref{eq:Z}, that $T = f(Z, Y)$ for some map $f$, and observe that $T_{\swap(\varsigma^{-1})} = f (Z_{\swap(\varsigma^{-1})}, Y)$. So by Lemma \ref{lem:exchange}, we have that $T_{\swap(\varsigma^{-1})}$ and $T$ are identically distributed. For any $k_j \in \{1, \ldots, \omega_j\}$ and $t_{j\ell} \in \real$ for $j = 1, \ldots, p$ and $\ell = 1, \ldots, \omega_j$, we have
  \begin{align}
    &\Prob\left[\bigcap_{j = 1}^p \{\kappa_j = k_j\}, \bigcap_{j = 1}^p \bigcap_{\ell = 1}^{\omega_j} \{T_{j, (\ell)} = t_{j\ell}\}\right] \nonumber\\
    = &\Prob \left[\bigcap_{j = 1}^p \{T_j^{(k_j)} = T_{j, (1)} = t_{j1}\}, \bigcap_{j = 1}^p \bigcap_{\ell = 1}^{\omega_j} \{T_{j, (\ell)} = t_{j\ell}\}\right] \nonumber\\
    =& \Prob \left[\bigcap_{j = 1}^p \{T_j^{(\varsigma_j^{-1}(k_j))} = T_{j, (1)} = t_{j1} \}, \bigcap_{j = 1}^p \bigcap_{\ell = 1}^{\omega_j} \{T_{j, (\ell)} = t_{j\ell}\} \right] \nonumber\\
    =& \Prob\left[\bigcap_{j = 1}^p \{\kappa_j = \varsigma_j^{-1}(k_j)\}, \bigcap_{j = 1}^p \bigcap_{\ell = 1}^{\omega_j} \{T_{j, (\ell)} = t_{j\ell}\}\right] \nonumber\\
    =& \Prob\left[\bigcap_{j = 1}^p \{\varsigma_j(\kappa_j) = k_j\}, \bigcap_{j = 1}^p \bigcap_{\ell = 1}^{\omega_j} \{T_{j, (\ell)} = t_{j\ell}\}\right] \nonumber,
  \end{align}
  where the first and the third equalities hold from the definition of $\kappa_j$'s, the second equality holds because $T_{\swap(\varsigma^{-1})}$ and $T$ are identically distributed, along with the fact that $(T_{\swap(\varsigma^{-1})})_{\mathrm{ordered}} = T_{\mathrm{ordered}}$. 
  Therefore, we have shown that 
  \begin{align}
    (\varsigma_1(\kappa_1), \ldots, \varsigma_p(\kappa_p), T_{\mathrm{ordered}}) \text{  and  } (\kappa_1, \ldots, \kappa_p, T_{\mathrm{ordered}}) \text{ are identically distributed}.
    \label{eq:joint}
  \end{align}
  For any $j \in \H_0$, now we further assume that $\varsigma_k$ is an identity permutation for all $k \neq j$, and $\varsigma_j$ is an arbitrary permutation on the set $\{1, \ldots, \omega_j\}$. The equality in joint distributions \eqref{eq:joint} implies that $\varsigma_j(\kappa_j)$ has the same distribution as $\kappa_j$. Since $\varsigma_j$ is an arbitrary permutation on the set $\{1, \ldots, \omega_j\}$, we have that $\kappa_j$ is uniformly distributed on the set $\{1, \ldots, \omega_j\}$, i.e., 
  \begin{align}
    \Prob(\kappa_j = i) = \omega_j^{-1} \qquad \forall i \in \{1, \ldots, \omega_j\}.
    \label{eq:marginal}
  \end{align}
  Furthermore, for any $i_k \in \{1, \ldots, \omega_k\}$ for $k \neq j$, and $t \in \real^{\sum_\ell \omega_{\ell}}$,
  \begin{align}
    \Prob\left[\varsigma_j(\kappa_j) = i \Big | \bigcap_{k \neq j} \left\{\kappa_k = i_k \right\}, T_{\mathrm{ordered}} = t \right] =&
    \frac{\Prob \left[\varsigma_j(\kappa_j) = i, \bigcap_{k \neq j} \left\{\varsigma_k(\kappa_k) = i_k\right\}, T_{\mathrm{ordered}} = t\right]}{\Prob \left[ \bigcap_{k \neq j} \left\{\kappa_k = i_k\right\}, T_{\mathrm{ordered}} = t \right]} \nonumber\\
    =& \frac{\Prob \left[\kappa_j = i, \bigcap_{k \neq j} \left\{\kappa_k = i_k \right\}, T_{\mathrm{ordered}} = t\right]}{\Prob \left[ \bigcap_{k \neq j} \left\{\kappa_k = i_k\right\}, T_{\mathrm{ordered}} = t \right]} \nonumber\\
    = &\Prob\left[\kappa_j = i \Big | \bigcap_{k \neq j} \left\{\kappa_k = i_k \right\}, T_{\mathrm{ordered}} = t \right],
    \nonumber
  \end{align}
  where the first equality holds from Bayes formula and the fact that $\varsigma_k$ is the identity permutation for all $k \neq j$, and the second equality holds from \eqref{eq:joint}. Therefore, for any $i_k \in \{1, \ldots, \omega_k\}$ for $k \neq j$, and $t \in \real^{\sum_\ell \omega_{\ell}}$, we have that
  \begin{align}
    \Prob \left[\kappa_j = i \Big | \bigcap_{k \neq j} \left\{\kappa_k = i_k \right\}, T_{\mathrm{ordered}} = t\right] = \omega_j^{-1} \qquad \forall i \in \{1, \ldots, \omega_j\}.
    \label{eq:conditional}
  \end{align}
  Combining \eqref{eq:marginal} and \eqref{eq:conditional}, we have that $\kappa_j$ is independent of $\{\kappa_k\}_{k \neq j}$ and $T_{\mathrm{ordered}}$.
\end{proof}

\section{Proof of Theorem \ref{thm:spotting}} \label{eq:proof_spotting}
Without loss of generality, we assume that the ordering in Step 3 of Section \ref{sec:mknockoffs} is such that $\sigma(j) = j$ for $j \in \{1, \ldots, p\}$. Consider 
\begin{align}
  \V (\R_k, c) = \frac{c^{-1}  + \sum_{j}  \indi\left \{ j \notin \R_k \right \}}{\left(\sum_{j} \omega_j \indi \left\{ j \in \R_k  \right\}\right) \vee 1} = \frac{c^{-1}  + \sum_{j = 1}^k  \indi\left \{ \kappa_j > 1 \right \}}{\left(\sum_{j = 1}^k \omega_j \indi \left\{ \kappa_j = 1  \right\} \right) \vee 1}
  \label{eq:wfdphat}
\end{align}
for some constant $c$. Recall that
\begin{align}
  \wfdp (\R_k) = \frac{\sum_{j} \omega_j \indi \left\{ j \in \H_0 \cap \R_k \right\} }{\left(\sum_{j} \omega_j \indi \left\{ j \in \R_k \right\}\right) \vee 1} = \frac{\sum_{j = 1}^k \omega_j \indi \left\{ j \in \H_0 \right\} \indi\left\{ \kappa_j = 1 \right\} }{\left(\sum_{j = 1}^k \omega_j \indi \left\{ \kappa_j = 1 \right\}\right) \vee 1}.
  \nonumber
\end{align}
We have the following key lemma:
\begin{lemma} \label{thm:key}
  Let $\V (\R_k, c)$ be defined as in \eqref{eq:wfdphat}. Then for any $\alpha \in (0, 1)$, there exists $x > 0$ such that
  \begin{align}
    \Prob \left[ \sup_k \frac{\wfdp(\R_k)}{\V (\R_k, c)} \geq x \right] \leq \alpha.
    \label{eq:thm}
  \end{align}
\end{lemma}

\begin{proof}[Proof of Lemma \ref{thm:key}]
  For any $x > 0$, from \eqref{eq:wfdphat},
  \begin{align}
    &\Prob \left\{ \sup_k \frac{\wfdp(\R_k)}{\V (\R_k, c)} \geq x \right\} \nonumber\\
    =& \Prob \left\{ \sup_k \left(\sum_{j = 1}^k \omega_j \indi \left\{ \kappa_j = 1 \right\} \indi \left\{ j \in \H_0 \right\} -  x \sum_{j = 1}^k  \indi \left \{ \kappa_j > 1 \right \}\right) \geq c^{-1}x  \right\} \nonumber\\
    \leq & \Prob \left\{ \sup_k \left(\sum_{j = 1}^k \omega_j \indi \left\{ \kappa_j = 1 \right\} \indi \left\{ j \in \H_0 \right\} -  x \sum_{j = 1}^k  \indi \left \{ \kappa_j > 1 \right \} \indi \left\{ j \in \H_0 \right\}\right) \geq c^{-1}x  \right\} \nonumber\\
    =& \Prob \left[ \sup_k \exp \left[ \theta \left\{ \sum_{j = 1}^k \omega_j \left( \indi \left\{ \kappa_j = 1 \right\} -  \frac{x}{\omega_j}  \indi \left \{ \kappa_j > 1 \right \} \right) \indi \left\{ j \in \H_0 \right\} \right\} \right] \geq \exp \left( c^{-1}x \theta \right)\right]
    \nonumber
  \end{align}
  for any $\theta > 0$. Define
  \begin{align}
    Z_k = \exp \left[ \theta \left\{ \sum_{j = 1}^k \omega_j \left( \indi \left\{ \kappa_j = 1 \right\} - \frac{x}{\omega_j}  \indi \left \{ \kappa_j > 1 \right \} \right) \indi \left\{ j \in \H_0 \right\} \right\} \right]
  \end{align}
  for $k \geq 1$, and $Z_0 = 1$.
  Next we find a value of $\theta > 0$ such that $\{Z_k\}$ is a super-martingale with respect to a certain filtration $\F_k$. If such a value of $\theta$ exists, then from Ville's maximal inequality for super-martingales \citep{ville1939etude}, we have that
  \begin{align}
    \Prob \left[ \sup_k \frac{\wfdp(\R_k)}{\V (\R_k, c)} \geq x \right] \leq \Prob \left\{ \sup_k Z_k \geq \exp(c^{-1}\theta x) \right\} \leq \frac{\E(Z_0)}{\exp(c^{-1}\theta x)} = \exp(-c^{-1}\theta x).
    \label{eq:ville}
  \end{align}

  So it is left to show that $Z_k$ is a super-martingale with respect to a filtration $\F_k$, where
  $\F_k$ is the $\sigma$-field generated from $\{\kappa_j\}_{j \leq k, j \in \H_0}$.
  First we observe that $Z_k$ is adapted to $\F_k$ for all $k$. By definition of a super-martingale, it is left to show that
  \begin{align}
    \E \left( \frac{Z_k}{Z_{k - 1}} \mid \F_{k - 1} \right) = \E \left[ \exp \left\{ \omega_k \theta \left( \indi \left\{ \kappa_k = 1 \right\} - \frac{x}{\omega_k} \indi \left\{\kappa_k > 1 \right\} \right) \indi \left\{ k \in \H_0 \right\} \right\} \mid \F_{k - 1} \right] \leq 1.
    \nonumber
  \end{align}
  First, we observe that this holds trivially for $k \notin \H_0$. For $k \in \H_0$, we have
  \begin{align}
    \E \left( \frac{Z_k}{Z_{k - 1}} \mid \F_{k - 1} \right) =& \E \left[ \exp \left\{ \omega_k \theta \left( \indi \left\{ \kappa_k = 1 \right\} - \frac{x}{\omega_k} \indi \left\{ \kappa_k > 1 \right\} \right) \right\} \mid \F_{k - 1} \right] \nonumber\\
    =& \E \left[ \indi \left\{ \kappa_k = 1 \right\} \exp \left( \omega_k \theta  \right) \mid \F_{k - 1}\right] + \E \left[ \indi \left\{ \kappa_k > 1 \right\} \exp \left(  -\theta x \right) \mid \F_{k - 1}\right] \nonumber\\
    =& \exp \left( \omega_k \theta \right) \Prob \left( \kappa_k = 1 \mid \F_{k - 1} \right) + \exp \left( - \theta x \right) \Prob \left( \kappa_k > 1 \mid \F_{k - 1} \right) \nonumber\\
    =& \frac{\exp \left( \omega_k \theta \right)}{\omega_k}  + \frac{(\omega_k - 1)\exp \left( - \theta x \right)}{\omega_k} \nonumber, 
  \end{align}
  where the last equality holds from Lemma \ref{lem:uniform}.

  For any fixed $\alpha \in (0, 1)$, take $x = {\theta}^{-1}(- c\log \alpha)$, which is equivalent to $\exp(-c^{-1}\theta x) = \alpha$. Then it remains to select $\theta$ such that for all $k \in \H_0$,
  \begin{align}
    \E \left( \frac{Z_k}{Z_{k - 1}} \mid \F_{k - 1} \right) = \frac{\exp \left( \omega_k \theta \right)}{\omega_k}  + \frac{\omega_k - 1}{\omega_k} \exp \left( c \log \alpha \right) \leq 1,
    \label{eq:toshow}
  \end{align}
  which is satisfied for
  \begin{align}
    \theta \leq  \frac{1}{\omega_k} \log \left\{ \omega_k - \left( \omega_k - 1 \right) \alpha^{c} \right\}.
    \nonumber
  \end{align}
  So we take
  \begin{align}
    \theta^\ast = \min_{k \in \H_0} \frac{1}{\omega_k} \log \left\{ \omega_k - \left( \omega_k - 1 \right) \alpha^{c} \right\}.
    \nonumber
  \end{align}
  Then \eqref{eq:toshow} holds and thus from \eqref{eq:ville}, the theorem holds with 
  \begin{align}
    x = \frac{-c\log \alpha}{\theta^\ast} = -c\log \alpha \left[\max_{k \in \H_0} \frac{\omega_k}{\log \left\{ \omega_k - \left( \omega_k - 1 \right) \alpha^{c} \right\}}\right].
    \label{eq:x}
  \end{align}
\end{proof}

Now we have
\begin{align}
  \U(\R_k, c) = x \V (\R_k, c) = -\log \alpha \left[\frac{1  + \sum_{j = 1}^k  c \indi\left \{ \kappa_j > 1 \right \}}{\left(\sum_{j = 1}^k \omega_j \indi \left\{ \kappa_j = 1  \right\}\right) \vee 1} \right] \left[\max_{k \in \H_0} \frac{\omega_k}{\log \left\{ \omega_k - \left( \omega_k - 1 \right) \alpha^{c} \right\}} \right],
  \nonumber
\end{align}
and the results in Theorem \ref{thm:spotting} follow.
\end{document}